\documentclass[prd,superscriptaddress,twocolumn,preprintnumbers,nofootinbib,%
  amsmath,amssymb,longbibliography]{revtex4-1}
  
\usepackage[breaklinks,colorlinks=true]{hyperref}
\usepackage[T1]{fontenc}
\usepackage[utf8]{inputenc}
\usepackage{mathtools}
\usepackage{amsmath,amsthm,amssymb}
\usepackage[dvipsnames]{xcolor}
\usepackage{bm}
\usepackage{dsfont}
\usepackage{lipsum}
\usepackage{calrsfs}
\DeclareMathAlphabet{\pazocal}{OMS}{zplm}{m}{n}

\usepackage{color}
\usepackage{hyperref}
\usepackage{bm}
\usepackage{xspace}
\allowdisplaybreaks
\usepackage{tikz}
\usetikzlibrary{matrix}

\definecolor{brickred}{rgb}{0.8, 0.25, 0.33}
\newcommand\myshade{85}
\colorlet{mylinkcolor}{BrickRed}
\colorlet{mycitecolor}{NavyBlue}
\colorlet{myurlcolor}{Aquamarine}

\hypersetup{
  linkcolor  = mylinkcolor!\myshade!black,
  citecolor  = mycitecolor!\myshade!black,
  urlcolor   = myurlcolor!\myshade!black,
  colorlinks = true,
}

\let\emptyset\varnothing

\newcommand\bX{\ensuremath{\bm X}\xspace}
\newcommand\bY{\ensuremath{\bm Y}\xspace}
\newcommand\PI[2]{\ensuremath{I_{\partial}^{#1\rightarrow #2}}\xspace}
\newcommand\IR[2]{\ensuremath{I_{\cap}^{#1\rightarrow #2}}\xspace}
\newcommand{\phiid}{\ensuremath{\Phi\mathrm{ID}}\xspace}

\begin{document}

\title{The Fast Möbius Transform: An algebraic approach to information decomposition}

\author{Abel Jansma}
\email{abel.jansma@mis.mpg.de}
\affiliation{Max Planck Institute for Mathematics in the Sciences, Leipzig}
\affiliation{School of Informatics, University of Edinburgh}

\author{Pedro A. M. Mediano}
\email{p.mediano@imperial.ac.uk}
\affiliation{Department of Computing, Imperial College London}
\affiliation{Division of Psychology and Language Sciences, University College London}

\author{Fernando E. Rosas}
\email{f.rosas@sussex.ac.uk}
\affiliation{Sussex AI and Sussex Centre for Consciousness Science, Department of Informatics, University of Sussex}
\affiliation{Principles of Intelligent Behavior in Biological and Social Systems (PIBBSS)}
\affiliation{Center for Psychedelic Research and Centre for Complexity Science, Department of Brain Science, Imperial College London}
\affiliation{Center for Eudaimonia and Human Flourishing, University of Oxford}

\newtheorem{definition}{Definition}
\newtheorem{conjecture}{Conjecture}
\newtheorem{theorem}{Theorem}
\newtheorem{lemma}{Lemma}
\newtheorem{proposition}{Proposition}
\newtheorem{corollary}{Corollary}
\newtheorem{example}{Example}
\newtheorem{remark}{Remark}

\begin{abstract}

\noindent
The partial information decomposition (PID) and its extension integrated information decomposition (\phiid) are promising frameworks to investigate information phenomena involving multiple variables. An important limitation of these approaches is the high computational cost involved in their calculation. Here we leverage fundamental algebraic properties of these decompositions to enable a computationally-efficient method to estimate them, which we call the \emph{fast Möbius transform}. Our approach is based on a novel formula for estimating the Möbius function that circumvents important computational bottlenecks. We showcase the capabilities of this approach by presenting two analyses that would be unfeasible without this method: decomposing the information that neural activity at different frequency bands yield about the brain's macroscopic functional organisation, and identifying distinctive dynamical properties of the interactions between multiple voices in baroque music. Overall, our proposed approach illuminates the value of algebraic facets of information decomposition and opens the way to a wide range of future analyses.

\end{abstract}

\maketitle

\section{Introduction}
Understanding how information is distributed and processed across complex systems is a fundamental challenge in various scientific domains, from neuroscience and genetics to machine learning and artificial intelligence~\cite{waldrop1993complexity,gleick2011information,bengio2024machine}. 
The \emph{Partial Information Decomposition} (PID) framework and its extension \emph{Integrated Information Decomposition} (\phiid) have emerged as powerful tools to address this question by decomposing the information provided by different variables into synergistic, redundant, and unique components~\cite{williams2010nonnegative,lizier2018information,mediano2021towards}. 
These frameworks have enabled a range of novel analyses revealing significant insights into, for example, the inner workings of natural evolution~\cite{rajpal2023quantifying}, genetic information flow~\cite{cang2020inferring}, psychometric interactions~\cite{marinazzo2022information,varley2022untangling}, cellular automata~\cite{rosas2018information,orio2023dynamical},
and polyphonic music~\cite{scagliarini2022quantifying}. 
Information decomposition have been particularly useful for studying biological~\cite{gatica2021high,luppi2022synergistic,varley2023information,varley2023multivariate,herzog2024high,pope2024time,varley2024emergence} and artificial~\cite{tax2017partial,proca2022synergistic,kaplanis2023learning,kong2023interpretable} neural systems, shedding light on the distinct roles of redundancy and synergy in neural computations~\cite{luppi2024information}.

An important limitation of PID and \phiid is the high computational cost involved in calculating the decomposition when increasing the number of involved variables. One of the most costly steps is the calculation of the information atoms by subtracting values over the so-called redundancy lattice. 
This operation relies on the Möbius inversion formula~\cite{williams2010nonnegative}, which provides a notion of discrete derivative 
over partially ordered sets~\cite{Rota1964}. 
Implementing the Möbius inversion requires knowledge of the \emph{Möbius function}, but the calculation of this function is in general highly non-trivial. 
In fact, the Möbius function on the redundancy lattice is currently unknown. 
Because of this, PID atoms are currently calculated by an alternative method that requires solving a system of equations, but the construction of this system quickly becomes unfeasible as the number of variables grows. 
More broadly, this knowledge gap prevents the information decomposition literature from accessing the rich mathematical toolbox associated with the calculus over lattices.

In this paper we present a novel approach that exploits the rich algebraic structure of the redundancy lattice, which leads to a method we call the \emph{fast Möbius transform}. The core of this method is an efficient way to compute the Möbius function of the redundancy lattice, which relies on algebraic properties of the lattice and the Birkhoff representation theorem~\cite{birkhoff1937rings}. This result allows us to compute information atoms without reconstructing the underlying lattice or solving a system of equations, enabling the calculation of PID and \phiid in larger systems than previously possible. 
We showcase the capabilities of this approach by presenting two analyses that would have been impossible without it: the decomposition of the information that the five frequency bands of the brain's electrical activity provide about the brain's macroscopic gradient of functional organisation; and the disentangling of distinctive information dynamics phenomena observed in music scores of J.S. Bach and A. Corelli.

The remainder of this paper is organised as follows. Section~\ref{sec:PID_lattice} reviews the principles behind the construction of PID and \phiid, and outlines relevant properties of functions defined over lattices. 
Section~\ref{sec:computational_adv} then introduces our main results, including the formula to calculate the Möbius function (Theorem~\ref{thm:MF_redundancy}) and approaches to efficiently compute certain atoms (Theorem~\ref{thm:synergyAtom}). 
Section~\ref{sec:computational_adv} presents the results of computational evaluations and case studies, and Section~\ref{sec:conclusion} outlines our main conclusions and discusses future work.

\section{Preliminaries} 
\label{sec:PID_lattice}

We start by introducing background information about PID and \phiid, and also about the calculus of functions defined over lattices. Throughout the manuscript, we denote partially ordered sets as calligraphic capital letters ($\mathcal{A}$), and antichains of partially ordered sets with bold greek letters ($\bm \alpha$).

\subsection{The partial information decomposition (PID)}\label{sec:introPID}

Partial information decomposition~\cite{williams2010nonnegative} builds on the intuition that the information that a set of predictor variables $\bm X=(X_1,\dots,X_n)$ contains about a target variable $Y$ (typically quantified by Shannon's mutual information $I(\bm X;Y)$) can be of different `types', and hence can be decomposed into various sub-components. For example, for the case of two predictor variables $\bm X=(X_1,X_2)$, PID states that the information $I(\bm X;Y)$ can be decomposed as follows
\begin{align}
    \nonumber I(\bm X; Y) =& I_\partial^{(1)(2)}(\bm X;Y) + I_\partial^{(1)}(\bm X;Y) + I_\partial^{(2)}(\bm X;Y) \\
    &+ I_\partial^{(12)}(\bm X;Y).
\end{align}
The terms on the right-hand side are referred to as information \textit{atoms}, and can be interpreted as follows: $I_\partial^{(i)}(\bm X;Y)$ is interpreted as the \emph{unique} information carried by $X_i$,  $I_\partial^{(1)(2)}(\bm X;Y)$ corresponds to the \emph{redundant} information that is carried by both $X_1$ and $X_2$, and $I_\partial^{(12)}(\bm X;Y)$ is the synergistic information carried only by the joint state of both variables together, but not by neither of them in isolation.

The general decomposition of the information provided by $n$ variables $\bm X$ about $Y$ is built following four steps:
\begin{itemize}
    \item[i)] Define as a `source' of information any subset of such variables, which are indexed as numbers within round brackets; for example, the source constituted by variables $X_1$, $X_3$, and $X_7$ is denoted by $(137)$. 
    \item[ii)] Establish a way to quantify the redundancy between sources, denoted by $I_\cap$, which captures the amount of information provided by multiple groups of variables; for example, the redundancy between sources $(123)$ and $(45)$ is captured by the quantity $I^{(123)(45)}_\cap(\bm X;Y)$. The redundancy of a source with a single group of variables is required to be equal to the corresponding mutual information, so that for example $I^{(12)}_\cap(\bm X;Y)= I(X_1,X_2;Y)$.
    \item[iii)] Assume that the redundancy between sources $S_1,\dots,S_n$ is the same if we remove any sources that are contained in other sources. For example, $I^{(12)(34)(4)}(\bm X;Y)=I^{(12)(34)}(\bm X;Y)$. This suggests to focus only on the redundancy between collections of sources that form an `antichain', where no source is contained into another. Also, it suggests that the redundancy should monotonically follow the following partial relationship: 
    \begin{equation}
        \bm \alpha \leq \bm \beta \quad \iff \quad \forall A\in\bm\alpha,\exists B\in\bm\beta, A\subseteq B. \label{eq:redundancy_ordering}
    \end{equation}
    The set of antichains of $n$ target variables is denoted by $A_n$, and together with the partial order $\leq$ constitutes the \emph{redundancy lattice} $\mathcal{R}_n = ( A_n,\leq)$. The Hasse diagrams of the redundancy lattice are shown in Figure \ref{fig:lattice_Red} for $n=2,3$ and $4$ predictor variables.
    \item[iv)] Finally, the PID atoms are defined as the `derivative' of $I_\cap^{\bm \beta}(\bm X; Y)$ on $\mathcal{R}_n$: i.e., via the  function $I_\partial^{\bm \beta}(\bm X; Y)$ that satisfies the following property for all  $\bm\alpha\in A_n$:
\begin{align}
    I_\cap^{\bm \alpha}(\bm X; Y) := \sum_{\substack{\bm \beta\in A_n\\\bm\beta\leq \bm\alpha}} I_\partial^{\bm\beta}(\bm X; Y). \label{eq:PID_sum}
\end{align}
    Using the PID atoms, the total information given by $\bm X$ about $Y$ is decomposed as
    \begin{equation}
    I(\bm X;Y) = \sum_{\bm \alpha\in A_n} I_\partial^{\bm\alpha}(\bm X;Y).
    \end{equation}
\end{itemize}

PID atoms are usually calculated by inverting a system of linear equations. For example, for the case of $n=2$ predictors, the PID formalism provides the following system of equations:
\begin{equation}
 \begin{cases}
    I(X_1, X_2; Y) =& I_\partial^{(1)(2)}(\bm X;Y) + I_\partial^{(1)}(\bm X;Y) \\
    &+I_\partial^{(2)}(\bm X;Y) + I_\partial^{(12)}(\bm X;Y), \label{eq:PID_underdetermined}\\
    I(X_1; Y) =& I_\partial^{(1)(2)}(\bm X;Y) + I_\partial^{(1)}(\bm X;Y), \\
    I(X_2; Y) =& I_\partial^{(1)(2)}(\bm X;Y) + I_\partial^{(2)}(\bm X;Y).
\end{cases}
\end{equation}
This provides three equations with four unknowns (the introduction of more predictor variables leads to larger but similarly underdetermined linear systems), and can be recursively solved by assuming a functional form for one of the PID atoms. Often a particular way to compute redundant information is adopted~\cite{williams2010nonnegative,griffith2014intersection}, which provides a value for $I_\partial^{(1)(2)}(\bm X;Y)$ and hence gives a unique solution for the other atoms. Alternatively, one can also define either the unique information~\cite{bertschinger2014quantifying,james2018unique} or synergy~\cite{rosas2020operational} and get the remaining atoms from it. 
Instead of addressing Eq.~\eqref{eq:PID_underdetermined} as a system of equations, in the next section we will investigate the substantial benefits of treating it via the proper Möbius function. 

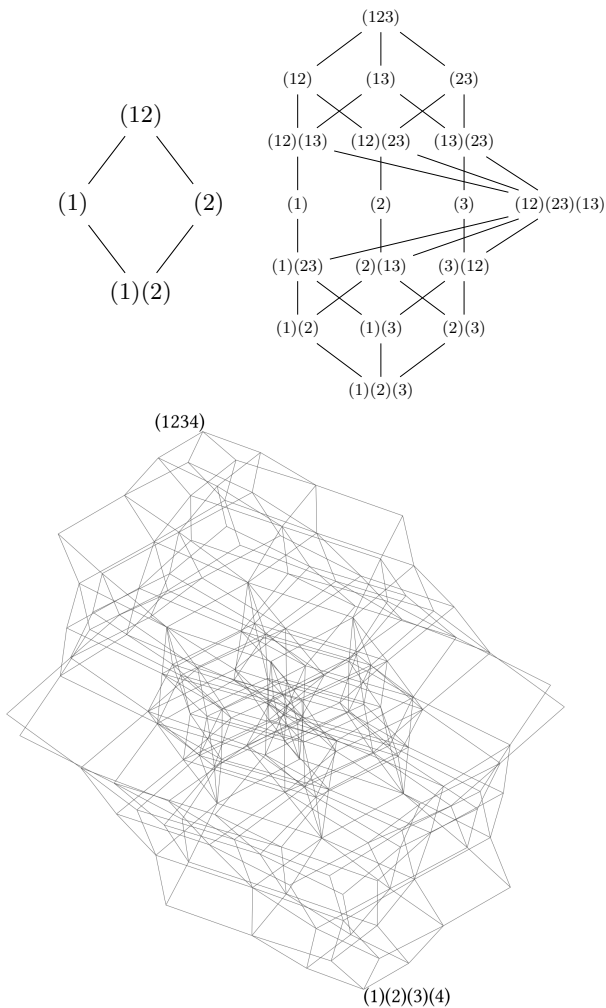
\begin{figure}
    \begin{minipage}[t!]{0.15\textwidth}
        \centering
        \begin{tikzpicture}[scale=0.7]
    \matrix (m) [matrix of math nodes, row sep=2em, column sep=0.2em] {
      & (12) & \\
      (1) & & (2) \\
      & (1)(2) & \\
    };
    \path[-]
      (m-1-2) edge (m-2-1)
              edge (m-2-3)
      (m-2-1) edge (m-3-2)
      (m-2-3) edge (m-3-2);
\end{tikzpicture}
  
    \end{minipage}
    \begin{minipage}[t!]{0.25\textwidth}
        \centering
        \scalebox{0.7}{
        \begin{tikzpicture}
    \matrix (m) [matrix of math nodes, row sep=2em, column sep=0.5em] {
      & (123) & &\\
      (12) & (13) & (23) & \\
      (12)(13) & (12)(23) & (13)(23)\\
      (1) & (2) & (3) & (12)(23)(13)\\
      (1)(23) & (2)(13) & (3)(12)\\
      (1)(2) & (1)(3) & (2)(3)\\
      & (1)(2)(3) & &\\
    };
    \path[-]
      (m-1-2) edge (m-2-1)
              edge (m-2-2)
              edge (m-2-3)
      (m-2-1) edge (m-3-1)
              edge (m-3-2)
      (m-2-2) edge (m-3-1)
              edge (m-3-3) 
      (m-2-3) edge (m-3-2)
              edge (m-3-3) 
      (m-3-1) edge (m-4-1)
              edge (m-4-4)
      (m-3-2) edge (m-4-2)
              edge (m-4-4)
      (m-3-3) edge (m-4-3)
              edge (m-4-4)
      (m-4-4) edge (m-5-1)
              edge (m-5-2)
              edge (m-5-3)
      (m-4-1) edge (m-5-1)
      (m-4-2) edge (m-5-2)
      (m-4-3) edge (m-5-3)
      (m-5-1) edge (m-6-1)
              edge (m-6-2)
      (m-5-2) edge (m-6-1)
              edge (m-6-3)
      (m-5-3) edge (m-6-2)
              edge (m-6-3)
      (m-6-1) edge (m-7-2)
      (m-6-2) edge (m-7-2)
      (m-6-3) edge (m-7-2);
\end{tikzpicture}
  }
    \end{minipage}
    \begin{minipage}[t!]{0.49\textwidth}
        \begin{flushleft}\includegraphics[width=0.9\columnwidth]{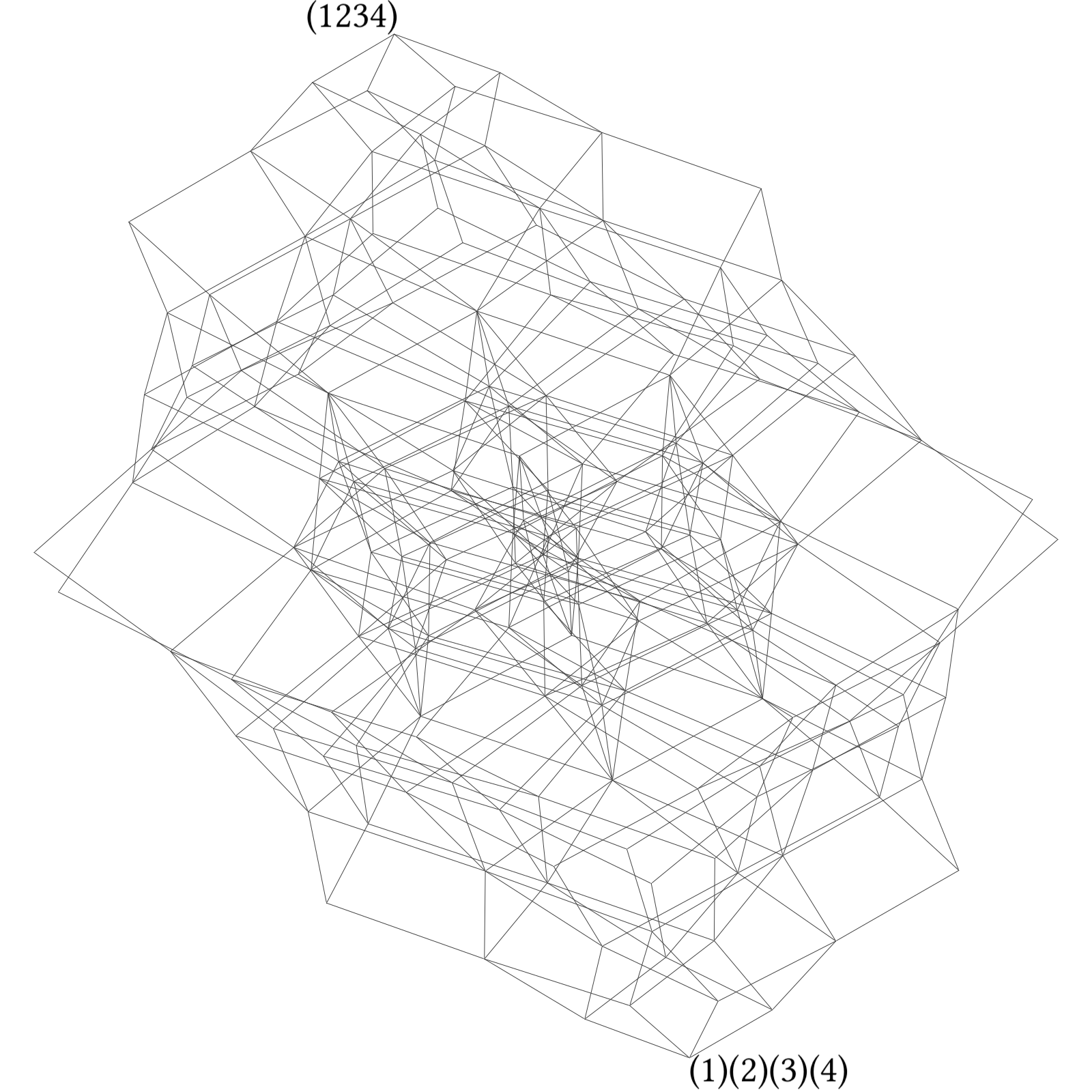}
        \end{flushleft}
    \end{minipage}
    \caption{The redundancy terms of a powerset of $n$ variables can be partially ordered: if $A$ and $B$ are redundancy terms, then $A\leq B$ if for every $b\in B$ there is an $a \in A$ such that $a\leq b$. Shown here are the transitive reduction (Hasse diagrams) of the redundancy lattices $\mathcal{R}_n$ for $n=2$ (top left), $n=3$ (top right middle) and $n=4$ (bottom, labels not shown) variables. Note that elements of $\mathcal{R}_n$ are the antichains in the powerset of $n$-variables. \label{fig:lattice_Red}}
\end{figure}

\subsection{Integrated information decomposition (\phiid)}

Integrated Information Decomposition (\phiid) is an extension of PID that allows one to decompose the information that a set of source variables $\bm X = (X_1, \dots, X_n)$ provides about a \textit{set} of target variables $\bm Y = (Y_1, \dots, Y_m)$~\cite{mediano2021towards}. 
While it can be applied to decompose the mutual information between any two sets of random variables, it has found natural applications in the analysis of multivariate time series data~\cite{luppi2022synergistic,rosas2020reconciling,mediano2022integrated}. This formalism can be built following similar steps as the PID:

\begin{itemize}
\item[a)] Define a `source' and `target' of information to be any subset of the $\bm X$ and $\bm Y$ variables, respectively. 
\item[b)] Establish a way to quantify the redundancy between sources and targets, which is denoted as $\IR{\bm\alpha}{\bm\beta}(\bm X;\bm Y)$ with $\bm \alpha$ and $\bm\beta$ collections of sources and targets, respectively. 
\item[c)] Assume that the redundancy between sources and targets is the same if we remove any sources $S$'s and targets $T$'s that are contained in other sources or targets, respectively, and hence focus on redundancy terms of the form $I_\cap^{\bm\alpha\rightarrow\bm\beta}(\bm X;\bm Y)$ where $\bm \alpha \in A_n$ and $\bm \beta \in A_m$ are antichains as in PID. These elements can be naturally ordered as
\begin{align}
\bm \alpha \rightarrow \bm \beta  \preceq \bm \alpha' \rightarrow \bm \beta' 
\quad \text{iff} \quad 
\bm \alpha \leq \bm \alpha'
\:\: \text{and} \: \:
\bm \beta \leq \bm \beta',
\end{align}
where $\leq$ denotes the ordering on $\mathcal{R}_n$. This gives rise to a lattice known as the \textit{double-redundancy lattice}~\cite{mediano2021towards}.

\item[d)] The \phiid atoms are obtained via the function $\PI{\bm\alpha'}{\bm\beta'}$ that satisfies 
\begin{align}
  \IR{\bm \alpha}{ \bm \beta}(\bm X;\bm Y) = \!\!\!\sum_{\substack{ \bm \alpha' \rightarrow\bm\beta' \preceq \bm\alpha\rightarrow\bm\beta}} \!\!\! \PI{\bm\alpha'}{\bm\beta'}(\bm X;\bm Y).
\end{align}
Then the joint mutual information between $\bm X$ and $\bm Y$ can be decomposed as
\begin{align}
    I(\bX; \bY) = \sum_{\substack{ \bm \alpha,\bm\beta \in  A}} \PI{\bm\alpha}{\bm\beta}(\bX; \bY) . \label{eq:phiID_sum}
\end{align}
\end{itemize}

Intuitively, \phiid is the `product' of two complementary PIDs, one decomposing the information carried by the sources about
the targets, and the other decomposing the information carried by the targets about the sources. In fact, the double-redundancy lattice can be shown to correspond to the cartesian product lattice $\mathcal{R}_n \times \mathcal{R}_n$ with the product order defined above, also known as Cartesian order~\cite[Sec.~2.15]{davey2002introduction}.

\subsection{Calculus over lattices and the Möbius transform}
\label{sec:calculus}
Sums over lattices, such as the information decompositions stated in Eq.~\eqref{eq:PID_sum} and \eqref{eq:phiID_sum}, can be interpreted as integrals over an interval of a partially ordered set. In fact, functions defined over partial orders possess a distinct calculus, of which we here present some fundamental concepts (more background about partially ordered sets is provided in Appendix~\ref{sec:prelims}).

An interval between $a$ and $b$ over a partially ordered set $\mathcal{R}$ can be naturally defined as $[a, b] := \{x\in\mathcal{R}: a\leq x \leq b\}$. Then, the `integral' of a function $f(x)$ between $a$ and $b$ over a partially ordered set $\mathcal{R}$ is defined as
\begin{equation}
    F(a,b) = \sum_{x\in[a,b]}f(x).
\end{equation}
To identify the derivative of a function (i.e. the operation that undoes integration), it is useful to first consider the set of all bi-variate functions as an algebra under the usual sum and multiplication given by the convolution
\begin{align}
    (F * G)(a, b) = \sum_{x\in[a,b]} F(a, x) G(x, b).
\end{align}
Crucially, in this setting integration can be understood as a particular type of convolution\footnote{Note that a univariate function $f(x)$ can be seen as bi-variate $f'(x,y) = f(x)$ considering a second variable $y$ with respect to which $f'$ is invariant.} given by
\begin{align}
    (\zeta * f)(a, b) 
    = \sum_{x\in[a,b]} \zeta(a,x) f(x)
    = \sum_{x\in[a,b]} f(x),
\end{align}
where $\zeta$ is the \emph{zeta-function} given by
\begin{align}
    \zeta(a, b) = \begin{cases}
        0 &\text{if } [a, b]=\emptyset,\\
        1 &\text{otherwise }.
    \end{cases}
\end{align}
When seen from this perspective, the `derivative' of a function can be defined as its convolution against the multiplicative inverse of the zeta function --- which is known as \emph{Möbius function}. This function was introduced by the mathematician A. F. Möbius in the 19-th century in the context of number theory~\cite{mobius1832besondere}, and was extended to partial orders by G-C. Rota in the 60's~\cite{Rota1964}.

Let us now establish the \emph{Möbius inversion formula} for partially ordered sets, which is analogous to the fundamental theorem of calculus but applied to partial orders. This result states that, if $f$ and $g$ are two functions defined over a partially ordered set $\mathcal{T}$, then the two following statements are equivalent~\cite{Rota1964}:
\begin{align}
    f(a) = \sum_{b \leq a} g(b) 
    \iff
    g(a) = \sum_{b \leq a} \mu(b, a) f(b),
\end{align}
where $\mu(b, a)$ is the Möbius function of the partial ordering $\mathcal{T}$, which can be recursively constructed via
\begin{align}
        \mu_\mathcal{T}(x, y) =
        \begin{cases}
            1 & \text{if } x=y,\\
            -\sum_{z\in[x,y]}\mu_\mathcal{T}(x, z) & \text{if } x < y , \\
            0 & \text{otherwise} .\label{eq:MF}
      \end{cases}       
\end{align}
If $\mathcal{T}$ is the natural numbers with their usual total order, then the Möbius inversion theorem recovers the discrete fundamental theorem of calculus with $g(x) = (f * \mu_{\mathbb{N}_\leq})(0, x) = f(x) - f(x-1)$ being the discrete derivative of $f$. In general the Möbius function takes a more complex form, which is determined by the chains in the underlying poset~\cite{stanley2011enumerative}. The functional turning $f$ into $g$ is known as the \emph{Möbius transform}.\footnote{Not to be confused with \emph{Möbius transformations}, which are rational functions on the complex plane.}

Overall, the Möbius transform shows that any sum over an interval in a partially ordered set can be inverted, as long as the corresponding Möbius function is known. This result is, indeed, the key that guarantees that PID and \phiid atoms can be identified from the redundancy function and the corresponding lattice structure. While calculating the Möbius function is often hard, special lattices allow a closed-form expression of their Möbius function: for example, if $\mathcal{T}$ is a powerset ordered by set inclusion, then its Möbius function is known to be $\mu(A, B) = (-1)^{|A| - |B|}$ \cite{stanley2011enumerative}. Our main result (Theorem~\ref{thm:MF_redundancy}) extends this finding to the PID and \phiid lattices.

\section{An efficient approach to compute information atoms}
\label{sec:computational_adv}

As described in the previous section, information decomposition atoms are usually calculated by recursively solving a system of equations. Here we exploit the structure of the Möbius inversion to identify more efficient ways to calculate PID and \phiid atoms.

\subsection{A free distributive lattice of redundancies and synergies\label{sec:FDL}}

The identification of the PID antichains $A_n$ and the construction of their partial ordering (via Eq.~\eqref{eq:redundancy_ordering}) is highly non-trivial, 
and this only becomes worse for \phiid. Here we present an alternative construction which uses algebraic tools to deliver the same redundancy lattice in a more intuitive manner. Our approach is to build an \textit{algebra of redundancies and synergies}, extending results reported in Refs.~\cite{finn2020generalised,gutknecht2021bits}.

We start from a set of $n$ primitive elements denoted by $(1),(2),\dots,(n)$, corresponding to $n$ sources of information, and two binary operations for building meets $\land$ and joins $\lor$. 
Of these, the join operation ($\lor$) can be thought of as denoting the union between sources and hence we interpret it as synergy, while the meet operation ($\land$) can be conceived as the intersection between them and hence we interpret as redundancy. For example, the atom associated with the antichain $(12)(3)$ can be thought of as obtained from the primitives $(1),(2)$ and $(3)$ via the algebraic operations $(1\lor2)\land(3)$, where $(1\lor2)$ relates synergy with the union of $(1)$ and $(2)$, and $(\cdot)\land (3)$ relates redundancy with the intersection of $(\cdot)$ and $(3)$.

The successive application of these two operations over the set of primitive elements give rise to an infinite number of formulae, which can be seen as a formal language~\cite{hopcroft2001introduction}.
Then, the second step is to impose properties over the operations to reduce the number of distinct expressions in this language. Specifically, we require the algebraic operations of synergy and redundancy to behave as unions and intersections do --- i.e., to be commutative, associative, idempotent, and obey the absortion laws (see Sec.~\ref{sec:algebra_of_lattices}). Additionally, we require redundancy and synergy to distribute over each other as
\begin{align}
    a \lor (b \land c) &= (a\lor b) \land (a \lor c),\\
    a \land (b \lor c) &= (a\land b) \lor (a \land c).
\end{align}
The set of the resulting distinct expressions is denoted by $D_n$. 
Note that the requirement that synergy and redundancy be commutative, associative, and idempotent is compatible with the original intuitions in partial information decomposition. One would expect, for instance, that $(1)(2)$, $(2)(1)$, and $(1)(1)(2)$ refer to the same collection of sources. The absorption laws also have a direct interpretation; for example $a \land (a \lor b) = a$ states that the information provided redundantly by $(a)$ and $(ab)$ is equal to the information given by $(a)$.

It can be shown that $D_n$ with the ordering implied by the meets and joins forms a lattice $\mathcal{F}_n=(D_n, \land, \lor)$ (see Sec.~\ref{sec:algebra_of_lattices}). This type of lattice is known as the \emph{free distributed lattice} over $n$ generators, as it is freely generated by primitive elements and its meets and joints distribute. Our next result shows that $D_n$ corresponds to the set of antichains $A_n$, and the lattice $\mathcal{F}_n$ is the same as the redundancy lattice $\mathcal{R}_n$ presented in Section~\ref{sec:introPID}. 
This result is a natural extension of the findings reported in Ref.~\cite[Sec.~5]{finn2020generalised} related to the properties of the union and intersection information, and the results presented in Ref.~\cite[Sec.~4]{gutknecht2021bits} about the relationship between PID antichains and logical propositions. 

\begin{proposition}
    Let $\mathcal{R}_n$ be the redundancy lattice on the set of antichains $ A_n$ ordered via $\leq$ as in \eqref{eq:redundancy_ordering}, and let $\mathcal{F}_n = (D_n, \land, \lor)$ be the free distributive lattice over $n$ elements, ordered as $a\preccurlyeq b \iff a\land b =a$. Then $ A_n=D_n$ and $\mathcal{R}_n = \mathcal{F}_n$. \label{lem:redundancy_is_FD}
\end{proposition}

\begin{proof}
    See Appendix.
\end{proof}

Building the redundancy lattice $\mathcal{R}_n$ as the free distributive lattice generated by the set of $n$ source variables can provide conceptual clarity. Indeed, this reveals the redundancy lattice to be a general algebraic construction, rather than dependent on a particular interpretation of redundant information, complementing the results reported in Ref.~\cite{gutknecht2021bits}. Additionally, understanding PID as arising from a distributive lattice explains why PID atoms can be visualised as intersections among circles in a Venn diagram --- while naive decomposition of entropy cannot. Indeed, this is a direct consequence of the fact that (finite) distributive lattices are isomorphic to lattices built from set intersections and unions, as shown by Birkhoff's representation theorem (see Appendix~\ref{app:proof_MF_redundancy}). See also related discussions presented in Ref.~\cite{finn2020generalised}.

Naturally, the same algebraic construction can be used to build \phiid's double redundancy lattice. For this, one needs to (i) construct the normal redundancy lattice $\mathcal{R}_n$, (ii) construct the Cartesian product of nodes, and finally (iii) identify the resulting partial ordering by evaluating it coordinate-wise. It is worth noticing that the product lattice is not in itself a free distributed lattice, but just a Cartesian product of two lattices that are. This, nonetheless, is enough to enable an efficient method to construct its Möbius function, as shown in the next subsection.

\subsection{The Fast Möbius Transform}
\label{sec:FMT}

The computation of PID atoms relies on the Möbius function, as a direct application of the Möbius inversion theorem (see Section~\ref{sec:calculus}) to the redundancy function $I_\cap^{\bm\alpha}(\bm X;Y)$ over the redundancy lattice $\mathcal{R}$ shows that 
\begin{align}
    I_\cap^{\bm \alpha}(\bm X; Y) &= \sum_{\bm \beta: \bm\beta\leq \bm\alpha} I_\partial^{\bm\beta}(\bm X; Y)\\
    &\Updownarrow\nonumber\\
    I_\partial^{\bm \beta}(\bm X; Y) &= \sum_{\bm \alpha:\bm\alpha\leq \bm\beta} \mu(\bm\alpha, \bm\beta) I_\cap^{\bm\alpha}(\bm X; Y) .\label{eq:MIT_PID}
\end{align}
Below we leverage tools of order theory to enable an efficient method to compute information decomposition atoms, which we call the \emph{Fast Möbius transform}.

Before presenting the main result behind the Fast Möbius transform, let us introduce some key ideas. Let us define the \emph{complement of an antichain} $\bm\alpha=\{a_1,\dots,a_l\} \in A_n$ via the mapping $*: \bm\alpha \mapsto \bm \alpha^* = \{\{1,\ldots,n\}\setminus a_j | a_j \in \bm \alpha\}$, where $A\setminus B$ is the set difference between $A$ and $B$. For example, if $n=3$ and $\bm\alpha = \{\{1, 2\}, \{2, 3\}, \{1, 3\}\}$, then $\bm\alpha^* = \{\{1\}, \{2\}, \{3\}\}$. Let us also introduce the notion of \emph{order ideal generated by the antichain} $\bm\alpha$, defined by
\begin{equation}
    I_{\bm\alpha} = \big\{ b \subseteq \{1,\dots,n\} \mid \exists a \in \bm\alpha \text{ s.t. } b \leq a\big\},
\end{equation}
corresponding to the set of all subsets of elements of $\bm\alpha$. For example, if $\bm\alpha=\{\{1,2\},\{3\}\}$ then $I_{\bm\alpha} = \{ \{1,2\},\{1\},\{2\},\{3\}, \emptyset\}$ (see Appendices~\ref{sec:prelims} and \ref{app:proof_MF_redundancy} for more background details). 
With these notions at hand, we can now present our next result.

\begin{theorem}[Fast Möbius transform for PID]
\label{thm:MF_redundancy}
    Let $P_N$ be the powerset of the set of indices $\{1,\dots,n\}$, and 
    $\mathcal{B}_n = (P_N, \subseteq)$ the powerset lattice ordered by inclusion. Then, the Möbius function on the redundancy lattice $\mathcal{R}_n$ can be calculated as
    {\small \begin{align}
        \mu_{\mathcal{R}_n}(\bm\alpha, \bm\beta) = \begin{cases}
            (-1)^{|I_{{\bm\alpha^*}} \!\setminus I_{{\bm\beta^*}}|}& \text{if $I_{{\bm\alpha^*}} \!\!\setminus I_{{\bm\beta^*}}$ is an antichain of $\mathcal{B}_n$}, \label{eq:MF_redundancy}\\
            0 & \text{otherwise.}
        \end{cases}
    \end{align}}
    Above, $I_{{\bm\alpha^*}}$ is the ideal in $\mathcal{B}_n$ generated by the complement of $\bm\alpha$ and $|I|$ is the cardinality of $I$.
\end{theorem}
We highlight that being an antichain of $\mathcal{B}_n$ is equivalent to being a PID atom (see Appendix~\ref{app:proof_MF_redundancy}).

\begin{proof}[Sketch of proof]
    The proof uses Birkhoff's representation theorem~\cite{birkhoff1937rings}, which states that any distributive lattice can be constructed by taking an adequate partially ordered set $\mathcal{T}$ and building the lattice $J(\mathcal{T})$ of all the ideals of $\mathcal{T}$ ordered by set inclusion. This construction is then combined with two facts: (i) that $J(\mathcal{B}_n)$ is isomorphic to the free distributed lattice generated by $n$ primitives, and (ii) that the Möbius function on a lattice $J(\mathcal{T})$ can be calculated as~\cite{stanley2011enumerative}
    \begin{align}
        \mu_{J(\mathcal{T})}(A, B) = \begin{cases}
            (-1)^{|B \setminus A|}& \text{if $B\!\setminus\!A$ is an antichain of $\mathcal{T}$},\\
            0 & \text{otherwise}.
        \end{cases}
    \end{align}
    To use this formula, all that is left is to explicitly construct an isomorphism that maps elements of $\mathcal{R}_n$ to the free distributed lattice generated by $n$ elements, which is done via the map $*$. 
    The full proof is presented in Appendix~\ref{app:proof_MF_redundancy}.
\end{proof}

\begin{corollary}[Fast Möbius transform for \phiid]
    The Möbius function on the double redundancy lattice $\mathcal{R}_n \times \mathcal{R}_n$ can be expressed as
    \begin{align}
        \mu_{\mathcal{R}_n \times \mathcal{R}_n} (\bm\alpha \rightarrow \bm\beta, \bm\alpha' \rightarrow \bm\beta') = \mu_{\mathcal{R}_n}(\bm\alpha, \bm\beta)\mu_{\mathcal{R}_n}(\bm\alpha', \bm\beta').
    \end{align}
\end{corollary}
\begin{proof}
    The result follows from the fact that the Möbius function on the product lattice with the Cartesian order is the product of the Möbius function on each of the lattices (see Ref.~\cite[Prop.~3.8.2]{stanley2011enumerative}).
\end{proof}

This result significant reduces the computational cost of computing information decomposition atoms. Currently, information atoms are typically calculated by inverting a system of equations (e.g. Eq.~\eqref{eq:PID_underdetermined} for the case of $n=2$ predictors), which is solved by starting at the lowest redundancy term and recursively solving for higher terms using
\begin{align}
    I^{\bm\alpha}_\partial(\bm X; Y) = I^{\bm \alpha}_\cap(\bm X; Y) - \sum_{\bm\beta < \bm\alpha} I^{\bm\beta}_\partial(\bm X; Y).
\end{align}
This amounts to solving on the order of $|\mathcal{R}_n|=|D_n|$ equations, where $|D_n|$ is the $n$-th Dedekind number. For example, the well-known Python package \texttt{dit}~\cite{dit} calculates information atoms using the following steps: (i) construct all antichains (time complexity of $\mathcal{O}(2^{2^n})$), (ii) order the antichains ($\mathcal{O}(|D_n|^2)$), (iii) calculate all redundancies ($\mathcal{O}(c(n)|D_n|)$, where $c(n)$ is a function that depends on the chosen redundancy measure), and finally (iv) recursively calculate the atoms ($\mathcal{O}(|D_n|^2)$). Overall, this approach scales with $\mathcal{O}(2^{2^n})$, and in practice has been found to be feasible only up to $n=3$ variables.

In contrast, when the Möbius function is known, it can be stored as a $|D_n| \times |D_n|$ matrix, and the calculation of PID atoms reduces to a matrix-vector multiplication:
\begin{align}
    I^{\bm\alpha}_\partial(\bm X; Y) &= 
    \left[\bm M^T \bm v\right]_\alpha,
\end{align}
where $[\bm M]_{\gamma \delta} = \mu_{\mathcal{R}_n}(\gamma, \delta)$ is the matrix that stores the Möbius function, and $[\bm v]_{\bm\gamma} = I_\cap^{\bm\gamma}(\bm X; Y)$. 
Crucially, while one still needs to know the antichains of $\mathcal{B}_N$, there is no need to order them. In addition, the antichains can be stored together with $\bm M$ at negligible cost. Therefore, this approach effectively bypasses the double exponential scaling when calculating the antichains, as well as the squared Dedekind scaling associated to ordering the lattice. This method thus allows one to exchange time complexity for space complexity, at the cost of calculating the Möbius function once. Moreover, the Möbius function is sufficiently sparse that it can be easily shipped with a package like \texttt{dit} for up to 5 variables, as shown in Section~\ref{sec:computationalEval}. Calculating any information atom then amounts to the inner product between a row of the Möbius function an the vector of redundancies ($\mathcal{O}(c(n)|D_n|)$). Therefore, while calculating all atoms thus still leads to a superexponential scaling $\mathcal{O}(c(n)|D_n|^2)$, we found that this approach is tractable for up to 5 variables (see~\ref{sec:computationalEval}). 

The approach developed here brings one additional benefit: it allows to compute some special atoms even without knowledge about the antichains and the Möbius function. In particular, our next result shows that the top-most synergy atom can be straightforwardly and efficiently calculated by exploiting the structure of the redundancy lattice --- but without explicitly building said lattice. 
\begin{theorem} \label{thm:synergyAtom}
    The top-most synergy atom, 
    $I_\partial^{(1\dots n)}(\bm X;Y)$, 
    can be directly expressed as a linear combination of the redundancies of all $(n-1)$-variable synergies as follows:
    \begin{align}
        I_\partial^{(1\dots n)}(\bm X; Y) = \sum_{U \subseteq \{1,\dots,n\}} (-1)^{n-|U|} I^{S_n^U}_\cap(\bm X; Y),
    \end{align}
    where $S_n^U = \{\{1, \ldots,  n\}\setminus x \mid x \in U\}$.
\end{theorem}

\begin{proof}
See Appendix \ref{app:proof_synergyAtom}.
\end{proof}

\vspace{0.2cm}

\section{Computational performance and case studies}

We now investigate the performance of the Fast Möbius transform from two angles: first by studying how it scales with the number of sources, and then by presenting results from two case studies that illustrate the capabilities of this approach for practical data analysis.

\subsection{Computational evaluations of PID and \phiid \label{sec:computationalEval}}

We investigated the performance of the fast Möbius transform (Theorem~\eqref{thm:MF_redundancy}) in calculating the complete Möbius function for the redundancy lattice $\mathcal{R}_n$ for different numbers of predictors $n$. 
The required computations were found to be relatively lightweight for $n=2,\ldots,5$, taking less than 60 minutes to calculate the $n=5$ case on a standard laptop. 
This allow us to efficiently use PID to study scenarios with up to five predictor variables, enabling analyses such as the one presented in Section~\ref{sec:brain_analyses}.

It is worth mentioning that the full 5-variable Möbius function could be stored in under 400kB, since only around 0.5\% of the possible entries are non-zero. This means that these values could be directly hardcoded into software packages featuring functions for information decomposition (e.g. \texttt{dit}), so it doesn't need to be calculated every run.  
That said, this approach may not yield useful results for more than $n=5$ information sources, as the Dedekind number for $n=6$ is $|D_6|=7,828,354$, which means that storing the 6-variable Möbius function is likely to take up over 400GB --- assuming similar sparsity. 
The resulting values of the Möbius function for various values of $n$ are shown in Figure \ref{fig:MF}. The matrices with the coefficients of the Möbius function for $n=2,\dots,5$ are available at a public repository.\footnote{\url{https://github.com/AJnsm/algebraicPID}}
\begin{figure}
    \centering
    \includegraphics[width=\columnwidth]{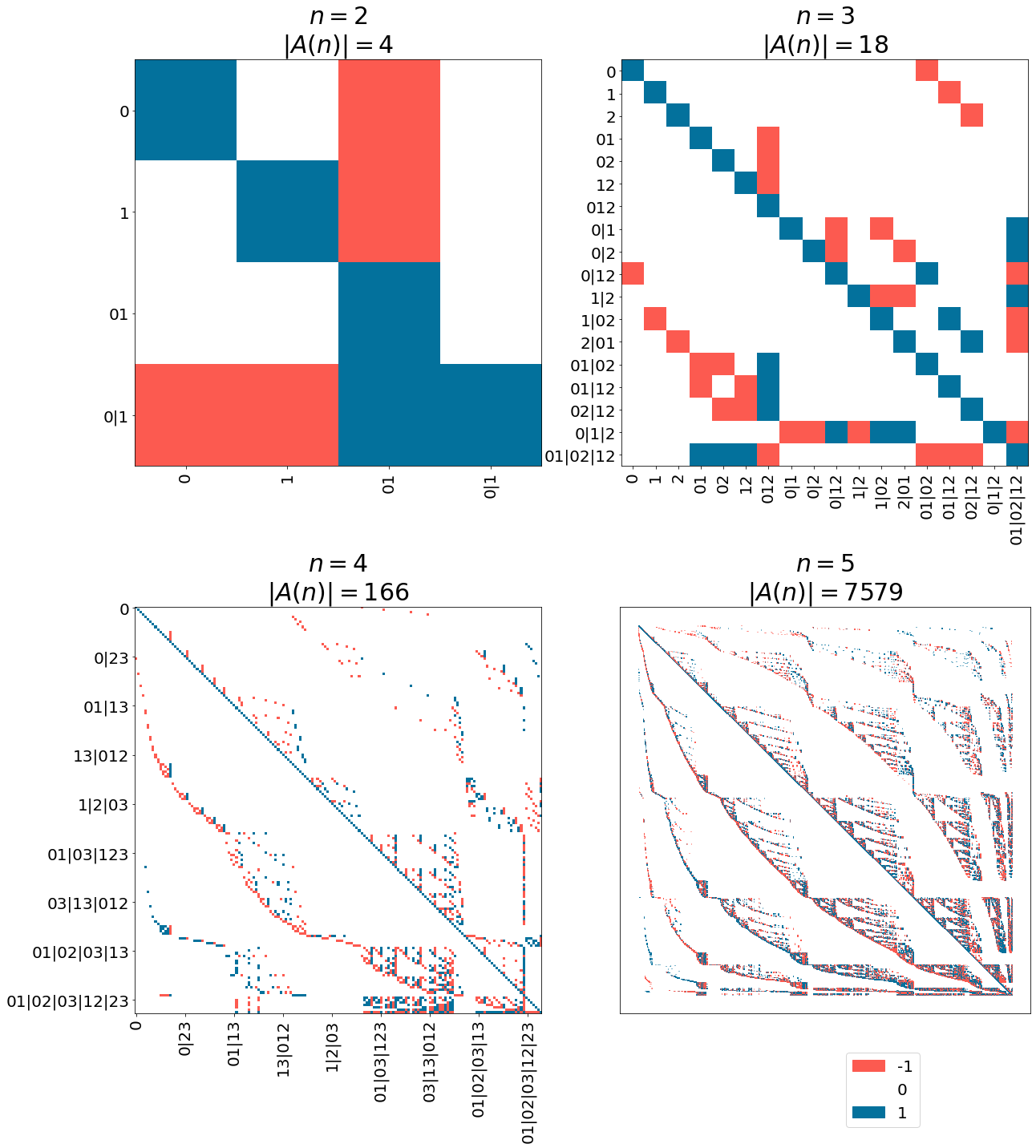}
    \caption{\textbf{The Möbius function of the redundancy lattice of up to 5 variables}. The rows and columns are ordered by the cardinality of the antichains. For $n=4$, only every 20th antichain is indicated along the axes. For $n=5$ none are shown.}
    \label{fig:MF}
\end{figure}

With the Möbius function calculated, we calculated the full 5-variable PID on four canonical distributions on 5 binary source variables ($X_1,\dots, X_5$) and a binary target variable $Y$:
\begin{align}
    p_{\text{uniform}}(\bm X, Y) &= \frac{1}{2^6},\\
    p_{\text{xor}}(\bm X, Y) &= \begin{cases}
        \frac{1}{2^5} & \text{if $Y = \sum_{i=1}^5 X_i \!\!\!\!\pmod{2}$},\\
        0 & \text{otherwise},
        \end{cases}\\
    p_{\text{unq}}(\bm X, Y) &= \begin{cases}
    \frac{1}{2} & \text{if $X_1 = Y$},\\
    0 & \text{otherwise},
        \end{cases}\\
    p_{\text{red}}(\bm X, Y) &= \begin{cases}
    \frac{1}{2} & \text{if $\forall i: X_i = Y$},\\
    0 & \text{otherwise}.
        \end{cases}
\end{align}
We then calculated the PID of these distributions using two well-known measures of redundancy: the measure  $I_\text{min}$ originally introduced by Williams and Beer \cite{williams2010nonnegative}, and the `minimal mutual information measure' $I_\text{mmi}$, introduced by the authors of \cite{bertschinger2013shared}. Our results, shown in Table \ref{tab:5VarPID}, confirmed that these two measures agreed in all cases and behaved exactly as expected. 
\begin{table}[h]
    \centering
    \caption{Value of specific PID atoms, $n=5$}
    \begin{tabular}{r|c|c|c}
         $I_\text{min/mmi}$ & $(12345)$ & $(1)(2)(3)(4)(5)$ & $(1)$\\
         \hline
         $p_\text{xor}$ & 1 & 0 & 0\\
         $p_\text{red}$ & 0 & 1 & 0\\
         $p_\text{unq}$ & 0 & 0 & 1 \\
         $p_\text{uniform}$ & 0 & 0 & 0 \\
    \end{tabular}
    \label{tab:5VarPID}
\end{table}

Finally, we also evaluated the capabilities of our approach to compute the Möbius inversion for \phiid with various numbers of time source and target variables. The current approach amounts to constructing two PID lattices and recursively solving all $|D_n|^2$ equations. Unfortunately, the largest number of variables this has been possible with thus far has been $n=2$. With the Möbius function approach, we can write the \phiid equations in terms of $\mu_{\mathcal{R}_n}$ as 
\begin{align}
    I^{\bm\alpha \rightarrow \bm \beta}_\partial(\bm X; Y) = \left[ (\bm M\otimes \bm M) \bm v\right]_{\bm\alpha \rightarrow \bm \beta},
\end{align}
where $\bm M$ is the matrix containing the Möbius function of the normal PID lattice (i.e., $[\bm M]_{\bm\gamma \bm\delta} = \mu_{\mathcal{R}_n}(\bm\gamma, \bm\delta)$), $\bm v$ is a vector encoding the redundancy function evaluated in the different \phiid atoms ($[\bm v]_{\bm\alpha \to \bm\beta} = I_\cap^{\bm\alpha \to \bm\beta}(\bm X, \bm Y)$), and $\otimes$ denotes the Kronecker product. Using the fact that $\bm M$ is sparse, the tensor product $(\bm M\otimes \bm M)$ can be calculated in under 20ms for the case of four predictors and four target variables, resulting in a matrix with 2,007,889 non-zero values (which is less that 0.3\% of all entries). Therefore, using this approach, it is straightforward to calculate the Möbius function of the \phiid lattice for up to four predictors and targets, enabling analyses as the one shown in Section~\ref{sec:baroque}.

\subsection{Decomposing the predictive power of the brain's frequency bands}
\label{sec:brain_analyses}

The human brain generates a complex repertoire of activity, commonly measured through techniques such as functional magnetic resonance imaging (fMRI) and electroencephalography (EEG). In fMRI, it is common to analyse correlations between the time-courses of activity across pairs of brain regions, summarised through the so-called \textit{functional gradient}~\cite{margulies2016situating}. In EEG, brain activity is conventionally analysed through its power spectrum and separated into five canonical frequency bands, referred to as $\delta, \theta, \alpha, \beta, \gamma$ in order of increasing frequency~\cite{buzsaki2004neuronal}. Comprehensively investigating synergy and redundancy among all five bands was previously an intractable problem. We here present the first complete decomposition of the information that these five frequency bands carry about the functional gradient of the human brain.

We obtained data describing the average functional gradient and spectral power (both in fsLR-4k surface-space) from the \texttt{neuromaps} library~\cite{markello2022neuromaps}, resulting in 8,000 samples (one per spatial location) each with 6 real numbers (5 frequency bands, to be used as sources, and the functional gradient, to be used as target). We quantified redundant information with the minimum mutual information (MMI) measure~\cite{barrett2015exploration} over binarised data based on its median value. The calculation of all 7,579 information atoms among the frequency bands using the fast Möbius transform took under a minute on a standard laptop. To determine statistical significance, we generated a null distribution for each atom through spatial null models (a.k.a. `spin tests') that preserve the spatial autocorrelation of the data~\cite{alexander2018testing,markello2021comparing}.

Our results reveal that 27 atoms (shown in Figure~\ref{fig:neuroPID}) were significant beyond the 95\textsuperscript{th} percentile of their null distribution.\footnote{Note that, as an exploratory analysis, these results were not corrected for multiple comparisons.} Interestingly, results suggest there are stark differences in the way different bands carry information about the functional gradient: for example, the $\delta$ band appears as an individual source by itself in 12 out of 27 of the significant atoms (more than any other band), whereas the $\theta$ and $\beta$ band never appear by themselves in any of the 27 atoms. This suggests that the $\delta$ band holds much of the redundant information that the brain's electrical activity carries about the functional gradient, whereas the $\theta$ and $\beta$ bands carry this information more synergistically. At the same time, 26 out of 27 atoms contain at least one synergistic source (i.e. one set with two or more bands), emphasising the pervasive role of synergy in linking local and global patterns of brain activity.

Although the neuroscientific implications of these results require further research, this proof of concept illustrates the possibilities enabled by the fast M\"obius transform in empirical data analyses.

\begin{figure}
    \centering
    \includegraphics[width=0.46\textwidth]{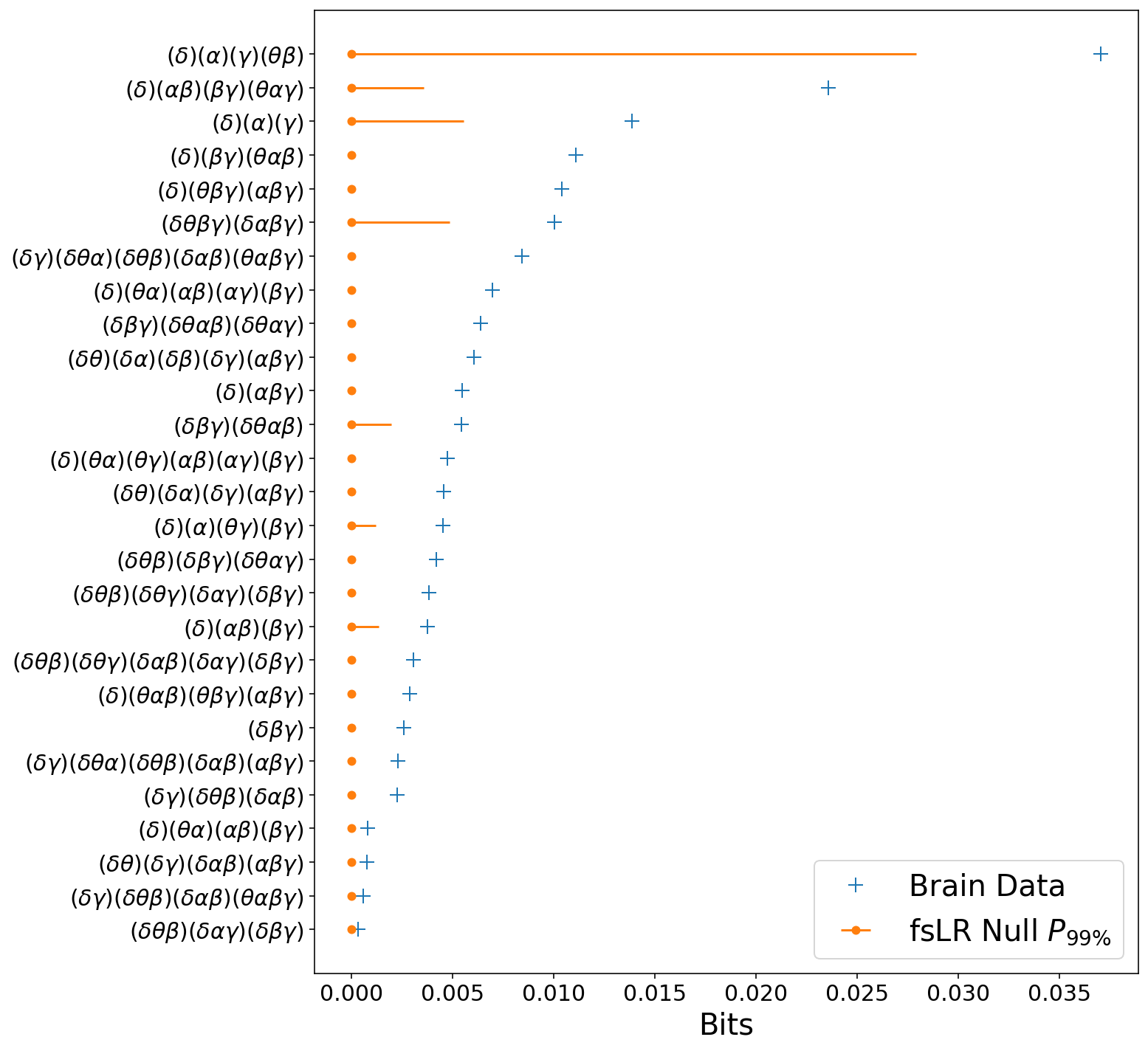}
    \caption{\textbf{Information about neural functional connectivity provided by the 27 strongest information atoms of the five canonical brain frequency bands}. Blue plus signs denote information atoms in the real brain data, and orange lines denote the 99\textsuperscript{th} percentile of a spatial null model (a.k.a. `spin test').}
    \label{fig:neuroPID}
\end{figure}

\subsection{Cross-scale interdependencies in baroque music \label{sec:baroque}}

We next turned to the $\Phi$ID framework, and used the Möbius function to investigate the relevance of high-order information dynamics found in music. 
Our analysis focuses on two sets of repertoire: the well-known chorales for four voices by Johann Sebastian Bach (1685-1750), and the music for string instruments Opus 1 and 3-6 by Arcangelo Corelli (1653-1713). 
These works are characterised by an elaborate interplay between the melodic lines and rich harmonic progressions, which results in a broad range of chords displayed along the repertoire. Furthermore, as typical in the Baroque
period (approx. 1600–1750), these pieces display a balance in the richness of each of the four voices, which contrasts with the subsequent Classic (1730–1820) and Romantic (1780–1910) periods where higher voices tend to take the lead while the lower voices provide support.

Following the preprocessing pipeline reported in Ref.~\cite{rosas2019quantifying}, we studied the chord transitions observed in music scores of Bach and Corelli containing four parts --- four voices (soprano, alto, tenor and bass) in the case of Bach’s chorales, and four string instruments (1st violin, 2nd violin, viola and cello) in the case of Corelli’s pieces. Every chord corresponds to a combination of four variables that can take values in an alphabet of 13 possible states (the 12 notes plus silence). Hence, chord transitions can be studied via the probability of observing two different subsequent chords. More details about prepossessing steps are provided in the Appendix.

\begin{figure}
    \centering
    \includegraphics[width=\columnwidth]{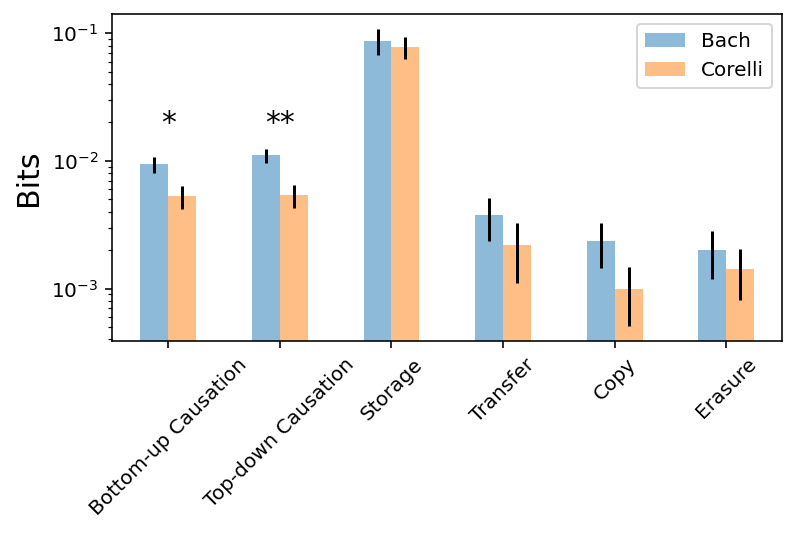}
    \caption{\textbf{Relevance of various categories of $\Phi$ID atoms for the Bach's and Corelli's music scores}. Result shows that the dynamics within the music of these two composers differ in the interplay between the whole and the parts, which is captured by the top-down and bottom-up causation categories. Significance was calculated via a t-test between the value of the atoms in each category, and significance is indicated as $p\leq 0.05$ with * and $p\leq \frac{0.05}{6}$ with ** (that is, ** indicates significance after Bonferroni correction). Error bars indicate standard error of the mean.}
    \label{fig:infoDyn_bach_v_corelli}
\end{figure}

We calculated the $\Phi$ID to investigate how chords carry information about the chord that follows them. For this, we consider $\bm X^4$ and $\bm Y^4$ to be the four music parts of two subsequent chords, and calculated the corresponding $166^2 = 27556$ $\Phi$ID atoms for the considered repertoire of both Bach and Corelli. 
Following Ref.~\cite{mediano2021towards}, we grouped these atoms within the following categories:
\begin{itemize}
    \item \textit{Bottom-up causation}: Any atom with a single target, where this target is larger than the largest source.
    \item \textit{Top-down causation}: Any atom with a single source, where this source is larger than the largest target.
    \item \textit{Storage}: Any atom where the source and the target are identical
    \item \textit{Transfer}: Any atom with only one source and one target, and where they are not the same. 
    \item \textit{Copy}: Any atom with a single source that is a proper subset of the target
    \item \textit{Erasure}: Any atom with a single target that is a proper subset of the source
\end{itemize}
Note that, in contrast to the $n=2$ case, these do not form completely disjoint sets for $n=4$. For example, the atom $(SAT)\to(B)$, that describes how the harmony between the Soprano, Alto, and Tenor determines the next melodic movement of the Bass, is both an instance of top-down causaution and transfer. 
To simplify the analysis, we estimated the relevance of each of these categories of atoms by calculating the mean value of the absolute value of the estimate of each atom involved. Each atom was bias-corrected by substracting the mean value obtained on a null distribution obtained by randomly shuffling the sequences of chords (see Appendix \ref{app:bach})).

Results show that Bach's and Corelli's scores exhibit significant differences related to their inter-scale dynamics, as reflected by higher levels of top-down and bottom-up causation in Bach's chorales (see Figure~\ref{fig:infoDyn_bach_v_corelli}). Other categories of \phiid atoms did not exhibited significant differences. These results were confirmed by performing additional control analyses, which are reported in the Appendix. Overall, these result suggest that Bach's chorales excel in how they allow for richer interplay between the whole and the parts, i.e., on how the activity of individual voices affect global chords and vice-versa.

\section{Conclusion \label{sec:conclusion}}

Here we introduced the fast Möbius transform as an algebraic method to efficiently calculate information decomposition. 
The fast Möbius transform is based on a closed-form formula for the Möbius function of the redundancy lattice, which turns the calculation of information atoms from the redundancy function into a simple matrix-vector multiplication --- opening the way to various numerical optimisation procedures (e.g. parallelisation, GPUs). 
Crucially, this method avoids the need to build the lattice itself or to invert a system of equations, which constitutes an important computational bottleneck to compute information decomposition. By doing this, the fast Möbius transform allows to run information decomposition analyses that were unfeasible before.

The practical utility of the fast Möbius transform was demonstrated in two case studies. 
First, we decomposed the information about the functional gradient observed in the human brain via fMRI provided by the various frequency bands observed in electroencephalographic recordings. Results reveal substantial differences in how different bands provide information: the $\delta$ band tends to provide information by itself, while $\theta$ and $\beta$ bands tend to carry information synergistically. 
Second, we disentangled the various information dynamics phenomena observed in the chorales of J.S. Bach, and contrast them against the ones observed in the music of his contemporary Arcangelo Corelli. Results show that Bach's music posses a stronger predominance of information modes involving phenomena across scales, where an individual voice affects global chords and vice versa.

The methods behind the fast Möbius transform open promising new ways to explore information decomposition from an algebraic perspective. 
In effect, the results presented here contribute to extend the mathematical foundations of information decomposition using the theory of finite distributive lattices and combinatorics. 
Indeed, showing that the redundancy lattice naturally emerges from a general algebraic construction --- that is independent of particular interpretations of the properties of redundant information --- supports the validity of this type of decompositions, while revealing their properties it from a new perspective. 
We anticipate that this algebraic perspective will open new avenues for the theoretical and applied study of multivariate information in complex systems. It should be emphasised that, unfortunately, computing the full PID and $\Phi$ID beyond $n=5$ are still difficult, regardless of how efficiently the Möbius function of the corresponding lattices can be calculated. Still, there might be situations in which only some of the atom decompositions are of interest, in which case the Möbius function for specific terms even though a full calculation is infeasible --- e.g. using results such as Theorem~\ref{thm:synergyAtom}.

More generally, our approach highlights the critical role of the Möbius inversion formula on the construction of information decomposition, following what is believed to be a general theme in the study of complex systems~\cite{jansma2023higher,jansma2024mereological}. 
We hope this work may motivate further investigations on the algebraic foundations of information decomposition, which may further enhance our computational tools while deepening our fundamental understanding of these frameworks.

\section*{Acknowledgements}

The authors thank Abdullah A. Makkeh for suggesting the name `fast Möbius transform'. A.J. is grateful to Jürgen Jost and Eckehard Olbrich for fruitful discussions on the partial information decomposition. FR was supported by the PIBBSS Affiliate program.

\appendix

\section*{APPENDIX}

\section{Background}
\label{sec:prelims}

\subsection{Orderings and lattices}
An `ordering', usually denoted by $\leq$, is a relation linking pairs of objects while observing three properties: 
\begin{itemize}
    \item[(i)] \textit{reflexivity}: $a\leq a$, 
    \item[(ii)] \textit{antisymmetry}: $a\leq b$ and $b\leq a$ if and only if $a=b$,
    \item[(iii)] \textit{transitivity}: $a\leq b$ and $b \leq c$ imply $a\leq c$.
\end{itemize}
When $a\leq b$ one says that $a$ is `less than or equal to' $b$. Furthermore, when $a\neq b$ and $a\leq b$, then we write $b>a$, and say that $a$ is `greater than' $b$. 

An example of an ordered set is the natural numbers together with their usual arithmetic ordering. 
Note that, in this case, any two elements $a,b\in\mathbb{N}$ are comparable (i.e. either $a\leq b$ or $b\leq a$), so $\leq$ is in called a \emph{total order}. 
However, the properties of orderings presented above allow for more general setups, in which any two elements $a, bS$ are either comparable (in which case either $a\leq b$ or $b \leq a$) or incomparable. 
Consider, for example, a collection of sets $S_i$ and an ordering $\leq$ under which $S_i \leq S_j$ if and only if $S_i \subseteq S_j$. This is usually not a total order, as there can be two sets where neither contains the other.
Such a relationship is called a \emph{partial ordering}, and a partially ordered set is called a \emph{poset}, being denoted by calligraphic letters as $(S, \leq)\coloneq \mathcal{S}$.

While total orderings arrange elements into a relatively simple structure in which all elements form a (potentially infinite) single chain $a\leq b\leq c \leq \dots$, posets can display richer structures. 
A poset $\mathcal{P}$ can be represented as a graph $G = (\mathcal{P}, E)$, where the set of edges $E$ contains an directed edge $(a, b) \in \mathcal{P}\times \mathcal{P}$ if and only if $a \leq b$. However, a more concise description is the \textit{transitive reduction} obtained by removing all relations that are implied by transitivity. Hence, the graph associated to the transitive reduction contains a directed edge $(a, b)$ if and only if $a\leq b$ and there is no $c \in \mathcal{P}$ such that $a \leq c$ and $c \leq b$. This graph is called the \emph{Hasse diagram} of the poset, and an example is given in Figure \ref{fig:hasseBA}. 

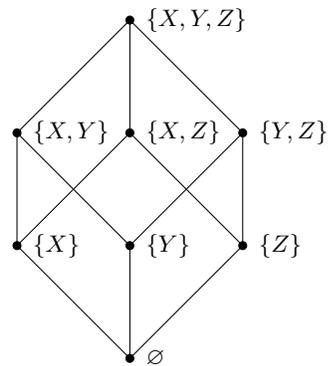
\begin{figure}
    \centering
    \begin{tikzpicture}[scale=1.5]

\coordinate[label={[xshift=1mm]right:$\{X,Y,Z\}$}] (111) at (0, 0);

\coordinate[label={[xshift=1mm]right:$\{X,Y\}$}] (110) at (-1, -1);
\coordinate[label={[xshift=1mm]right:$\{X,Z\}$}] (101) at (0, -1);
\coordinate[label={[xshift=1mm]right:$\{Y,Z\}$}] (011) at (1, -1);

\coordinate[label={[xshift=1mm]right:$\{X\}$}] (100) at (-1, -2);
\coordinate[label={[xshift=1mm]right:$\{Y\}$}] (010) at (0, -2);
\coordinate[label={[xshift=1mm]right:$\{Z\}$}] (001) at (1, -2);

\coordinate[label={[xshift=1mm]right:$\emptyset$}, ] (000) at (0, -3);

\foreach \a/\b in {111/110, 111/101, 111/011, 110/100, 110/010, 101/100, 101/001, 011/010, 011/001, 100/000, 010/000, 001/000}
    \draw (\a) -- (\b);

\foreach \v in {000,001,010,011,100,101,110,111}
    \draw[fill=black] (\v) circle (1pt);

\end{tikzpicture}
    \caption{The powerset $\mathcal{P}(\{X, Y, Z\})$ of a set of $3$ variables, ordered by set inclusion, forms a poset known as a Boolean algebra. Shown here is the transitive reduction (Hasse diagram) of the Boolean algebra on 3 variables. It is common to keep the directionality implicit by drawing the diagram such that all edges are oriented downwards. Note that $\{X\} \subseteq \{X, Y, Z\}$, but that the Hasse diagram contains no such edge because $\{X\} \subseteq \{X, Y\} \subseteq \{X, Y, Z\}$. For arbitrary $n$, the Hasse diagram of a Boolean algebra describes an $n$-cube.}
    \label{fig:hasseBA}
\end{figure}

A poset $\mathcal{P}$ allows one to define upper and lower bounds. Given a subset $S\subseteq \mathcal{P}$, the element $b\in \mathcal{P}$ is a lower bound for $S$ if $p\leq s$ for all $s\in S$. The greatest lower bound is called the \textit{meet} of S. Similarly, one could define the least upper bound of $S$, referred to as its \textit{join}. 
If $S=\{a,b\}$, we denote the join of $S$ as $a \land b$, and a meet as $a \lor b$.

A poset where all such meets and joins exist is called a \textit{lattice}. Please note that not all partial orderings are lattices, as meets and joins are not guaranteed to always exist. Finite lattices have a unique element that is larger than all others known as the \textit{supremum}, and a unique element smaller than all others known as \textit{infimum} --- therefore, their graph representation always looks like a diamond with a single element in top and bottom (for an example, see Figure~\ref{fig:hasseBA}).

\subsection{Lattices as algebraic structures}
\label{sec:algebra_of_lattices}

The above construction establishes lattices in terms of properties of a partial order. Interestingly, there is an alternative way to build lattices in algebraic terms, which we expose next.

Given a set $P$, let us introduce two operations $\land$ and $\lor$ that take two elements of $P$ as inputs and deliver a third element of $P$ while satisfying three basic properties:
\begin{enumerate}
    \item \textit{Commutative}: $a\land b=b\land a$ and $a\lor b=b\lor a$,
    \item \textit{Associative}: $(a\land b)\land c= a\land (b\land c)$ and $(a\lor b)\lor c= a\lor (b\lor c)$,
    \item \textit{Idempotent}: $a\land a=a$, $a\lor a=a$. 
\end{enumerate}
In addition to this, let's require these operations to interact with each other via the following laws: 
\begin{enumerate}
    \item[4.] \textit{Absorption}: $a \land (a \lor b) = a$ and $a \lor (a \land b) = a$.
\end{enumerate}
These are the same properties that one would naturally require from the operations of union and intersection of sets, and from the logical operations \texttt{and} and \texttt{or}.

It can be shown that the algebraic structure that these two operations induce on the underlying set $P$, written as $\mathcal{L} = (P, \land, \lor)$, is precisely a lattice in the order-theoretic sense described above with the order relationship $a \leq b \iff a \land b = a$. 
In other words, an ordering on $P$ is equivalent to an algebraic structure of joins and meets.

The equivalence between lattices and this algebraic construction is important, because it opens the door to study ordering relationship with the rich toolkit that can be found in the algebraic literature. 
One important example is the case of Boolean lattices, which require a third operation equivalent to a negation. 
Another important example is distributed lattices, where the operations $\land$ and $\lor$ are required to distribute over each other.

Another interesting consequence of the link between orderings and algebras is the possibility to build lattices as reflecting a formal language via free algebras. Free algebras are build starting from a number of primitive expressions (say $a,b,\dots$), and then explore all possible algebraic combinations of them. For example, if $\land$ and $\lor$ distribute, then the formal language resulting from a free algebra over two primitives $a$ and $b$ only has four elements: $a$, $b$, $a\land b$, and $a\lor b$, which has a natural lattice structure considering $\land$ and $\lor$ as meets and joins. Such lattices are discussed in more detail in Section \ref{sec:FDL}.

\section{Proofs}
\subsection{Proof of Proposition~\ref{lem:redundancy_is_FD}} \label{proof:redundancy_is_FD}
\begin{proof}
    Since the operations $\land$ and $\lor$ of $\mathcal{F}_n$ are commutative, associative, and idempotent, each element of this lattice can be reduced to the form $S_1\land S_2 \land S_3 \land \ldots$, where the $S_i$ do not contain $\land$ operators and are mutually incomparable (since $A\land B = A$ if $A\leq B$). That is, the lattice is an ordering on the set of antichains of $\mathcal{P}(N)$. Left to show is that this ordering is the same. 

    Given two antichains $A$ and $S=S_1\land S_2 \land S_3 \land \ldots$, assume that $A\leq S$, such that $\forall S_i \in S~ \exists a \in A: a \subseteq S_i$. Now, note that for any two elements of $\mathcal{F}_n$, if $a\subseteq b$, then $b=a \lor c$ for some $c\in \mathcal{F}_n$. By the absorption laws, one can find that $a\land b = a\land (a \lor c)=a$. We write 
    \begin{align}
        A\land S = \left(\bigwedge_{a \in A}a \right)\land \left(\bigwedge_{S_i \in S} S_i\right).
    \end{align} 
    Because $\land$ and $\lor$ are associative and commutative, one is free to rearrange these terms. By assumption, each $a\in A$ can be paired with an $S_a\in S$ such that $a\subseteq S_a$. $A\land S$ can thus be written as $a_1 \land a_2 \land \ldots \land a_n \land S_1 \land S_2 \land \ldots \land S_n$, and terms can be freely reordered so until somewhere a term $a_i\land S_i$ appears for which $a_i\subseteq S_i$. This reduces to $a_i$, and this procedure can be done until all $S_j$ are eliminated from the expression. This shows that ${A\leq S \implies A\land S = A} \iff A\preccurlyeq S$.

    Left to show is the reverse implication. Assume that $A\land S = A$, but that there is an $S_i \in S$ such that $\nexists a\in A$ is $a\subseteq S_i$. This means that the reduction procedure from above can be performed, until one is left with $A\land S = A\land S_i$. For this to be equal to $A$ for all choices of $A$ and all ordering of elements of $A$, $S_i$ must be the identity for the $\land$ operation. This means that $S_i$ is the maximal element of $\mathcal{F}_n$, which is the set $N$ itself. However, that means that \textit{every} $a\in A$ is a subset of $S_i$, which is in contradiction with the assumption that there is no such $a$. This shows that $A \preccurlyeq S \iff A\land S = A \implies A\leq S$, and completes the proof. 
\end{proof}

\subsection{Proof of Theorem \ref{thm:MF_redundancy}} \label{app:proof_MF_redundancy}
This proof requires some preliminary definitions and lemmas. First we review a general construction that can generate a distributive lattice $\mathcal{L}$ from the ideals of another partially ordered set $\mathcal{P}$, known as Birkhoff's representation theorem. This result is so ubiquitous that it is often referred to as the `fundamental theorem of finite distributive lattices.'

Before stating this theorem, let's introduce the notion of \emph{ideal} in the context of order theory. For a given poset $\mathcal{P}$, one can define an ideal as a subset $I\subseteq \mathcal{P}$ such that for all $x\in I$ and $y\in \mathcal{P}$ then
\begin{align}
    y\leq x \iff y \in I.
\end{align}
The set of all ideals of $\mathcal{P}$ can be ordered by inclusion, and forms a lattice denoted by $J(\mathcal{P})$.

\begin{theorem}[Birkhoff \cite{birkhoff1937rings}]
    Let $\mathcal{L}$ be a finite distributive lattice. Then there exists a a unique (up to isomorphism) partially ordered set $\mathcal{P}$ such that $\mathcal{L}$ is isomorphic to the lattice of ideals $J(\mathcal{P})$. 
\end{theorem}

Given $N=\{1,\dots,n\}$ a set of $n$ indices, let $\mathcal{B}_n$ denote the power set of $N$ ordered by inclusion. Note that, in general, the collection of maximal elements of each ideal within $J(\mathcal{P})$ is equal to the set of antichains of $\mathcal{P}$. Thus, thanks to Birkhoff's result, $J(\mathcal{B}_n)$ is guaranteed to be a distributive lattice over the antichains of $\mathcal{B}_n$. More specifically, it can be shown that $J(\mathcal{B}_n)$ is the free distributive lattice over $n$ generators \cite{stanley2011enumerative}. 

One convenient property of distributive lattices constructed as ordered ideals $J(P)$ is that the Möbius function on $J(\mathcal{P})$ can be expressed as \cite{stanley2011enumerative}:
\begin{align}\label{eq:MF_distributive}
    \mu_{J(\mathcal{P})}(A, B) = \begin{cases}
        (-1)^{|B \setminus A|}& \text{if $B\setminus A$ is an antichain of $\mathcal{P}$},\\
        0 & \text{otherwise}.
    \end{cases}
\end{align}
This property would give us a convenient way to calculate the Möbius function of the redundancy lattice, but unfortunately it cannot be used directly as 
the set of maximal ideals $J(\mathcal{B}_n)$ do not directly give $\mathcal{R}_n$. For example, for $N=\{0, 1, 2\}$ the largest ideal (that is not $\mathcal{B}_n$ itself) has as supremum the element $\{\{0, 1\}, \{0, 2\}, \{1, 2\}\}$, which is not the maximal term on the redundancy lattice (i.e. $\{\{0, 1, 2\}\}$). To instantiate the correspondence, we explicitly construct the isomorphism between the redundancy lattice and $J(\mathcal{B}_n)$.

\begin{lemma}\label{lem:red_isomorphism}
    Let $J(\mathcal{B}_n)$ be the free distributive lattice over $n$ generators, and let $\mathcal{R}_n$ be the redundancy lattice of $n$ variables. Then $J(\mathcal{B}_n)$ and $\mathcal{R}_n$ are isomorphic. Furthermore, the anti-isomorphism is realised by a map that sends a term in $\mathcal{R}_n$ to the ideal generated by its complement. 
\end{lemma}
\begin{proof}
    Let $ A_n$ be the set of antichains in of $\mathcal{B}_n$, and $*$ the map considered in Section~\ref{sec:FMT} (i.e., $*: A_n \to A_n$ by taking the complement of each element of $\bm \alpha$ as $*: \bm \alpha \mapsto \bm \alpha^* = \{N\setminus a | a \in \bm \alpha\}$). 
    Note that $\bm \alpha\subseteq \bm \beta \iff \bm \beta^* \subseteq \bm \alpha^*$. 
    Under the map $*$, the ordering of redundancy terms can be rewritten as
    \begin{align}
        \bm\alpha \leq \bm\beta \iff \forall b^* \in \bm\beta^* ~ \exists a^* \in \bm\alpha^* \text{ s.t. } b^* \subseteq a^* .\label{eq:star-redundancy-ordering}
    \end{align}

    Let $I_{\bm\gamma} = \{ t \in \mathcal{B}_n \mid \exists x \in \bm\gamma \text{ s.t. } t \leq x\}$ be the ideal in $J(\mathcal{B}_n)$ generated by $\bm\gamma \in A_n$. Equivalently, $I_{\bm\gamma}$ is the ideal of $\mathcal{B}_n$ with $\gamma$ as its maximal element. We can then write
    \begin{align}
        I_{\bm\alpha^*} &= \{ t \in \mathcal{B}_n \mid \exists a^* \in \bm\alpha^* \text{ s.t. } t \subseteq a^*\},\\
        I_{\bm\beta^*} &= \{ t \in \mathcal{B}_n \mid \exists b^* \in \bm\beta^* \text{ s.t. } t \subseteq b^*\}.
    \end{align}
    Now since $J(\mathcal{B}_n)$ is simply the set of all ideals of $\mathcal{B}_n$ ordered by inclusion, i.e. $I \preccurlyeq I' \iff I \subseteq I'$, we can write the ordering on $J(\mathcal{B}_n)$ as 
    \begin{align}
        I_{B^*} \preccurlyeq I_{A^*} \iff \forall x:~ (x \in I_{B^*} \implies x \in I_{A^*}).\label{eq:star-ideal-ordering}
    \end{align}
    In words, both \eqref{eq:star-redundancy-ordering} and \eqref{eq:star-ideal-ordering} express that any element of $B^*$ is the subset of an element in $A^*$, so the two orderings are equivalent. This means that $A\leq B \iff I_{B^*} \preccurlyeq I_{A^*}$, and so the map $A \mapsto I_{A^*}$ is an anti-isomorphism between the redundancy lattice $\mathcal{R}_n$ and $J(\mathcal{B}_n)$. Since $J(\mathcal{B}_n)$ is itself isomorphic to its reverse order (by replacing subsets by supersets in the ordering, or reversing the implication in Equation \eqref{eq:star-ideal-ordering}), we have that $\mathcal{R}_n$ is isomorphic to $J(\mathcal{B}_n)$.
\end{proof}

With all these elements in place, we can now present the proof of the Theorem.

\begin{proof}[Proof of Theorem~\ref{thm:MF_redundancy}]
    The Möbius function of $J(\mathcal{B}_n)$ is given by Equation \eqref{eq:MF_distributive}. Since the free distributive lattice is self-dual, Lemma \ref{lem:red_isomorphism} implies that the Möbius function of $\mathcal{R}_n$ is given by the same expression, after each antichain $\bm\alpha$ to its complement-ideal $I_{{\bm\alpha^*}}$.
\end{proof}

This construction can be summarised as follows. Let ${i:  \mathcal{R}_{n} \to J(\mathcal{B}_{n})}$ map each antichain to the ideal it generates as ${i: A \mapsto I_A}$. Let $\mu_{\mathcal{R}_n}$ and $\mu_{J(\mathcal{B}_n)}$ be the Möbius functions on $\mathcal{R}_n$ and $J(\mathcal{B}_n)$, respectively. Then the following diagram commutes:

\begin{tikzpicture}[auto]
    \node (A)  at (0,0) {$\mathcal{R}_n \times \mathcal{R}_n$};
    \node (B)  at (5,0) {$J(\mathcal{B}_n) \times J(\mathcal{B}_n)$};
    \node (C)  at (5,-3) {$\mathbb{Z}$};
    
    \draw[->] (A) to node {$(i\circ *) \times (i\circ *)$} (B);
    \draw[->] (A) to node [swap] {$\mu_{\mathcal{R}_n}$} (C);
    \draw[->] (B) to node {$\mu_{J(\mathcal{B}_n)}$} (C);
\end{tikzpicture}

\subsection{Proof of Theorem \ref{thm:synergyAtom} \label{app:proof_synergyAtom}}
\begin{proof}
    The proof splits into three parts. i) We show that the redundancies among $(n-1)$-variable synergies form a Boolean algebra at the top of the redundancy lattice. ii) We show that any other element in the lattice is not part of a Boolean algebra at the top of the lattice. iii) We show that this leads to the above expression.

    i) Consider an element of $\mathcal{R}_n$ that corresponds to the redundancy among $(n-1)$-variable synergies. It is of the form
    \begin{equation}
        \{\{\bm X\setminus x\} \mid x \in U\} \label{eq:antichain_at_top}
    \end{equation}
    for some $U\subseteq \bm X$. Under the bijection $i\circ *$, this is sent to a term
    \begin{align}
        I_{\{\{x\} \mid x \in U\}} = \{\{x\} \mid x \in U\} \cup \emptyset. \label{eq:ideal_at_bottom}
    \end{align}
    This bijection is an anti-isomorphism of the redundancy lattice when the subsequent ideals are ordered by inclusion, which means that the terms of the form \ref{eq:antichain_at_top} form a Boolean algebra if and only if the terms of the form \ref{eq:ideal_at_bottom} do. Note that ordering the latter simply corresponds to ordering by inclusion on all different $U\subseteq \bm X$, so that they form a Boolean algebra at the bottom of $J(\mathcal{B}_n)$. Therefore, the redundancies among $(n-1)$-variable synergies form a Boolean algebra at the top of $\mathcal{R}_n$. This proves i). 

    ii) We use the fact that the free distributive lattice is isomorphic to the lattice of monotonic Boolean functions (functions of the type $f: \mathcal{B}_n \to \{0, 1\}$ such that $f(A) \leq f(B) \iff A \leq B$). Note that these functions are fully defined by the set $U\subset \mathcal{B}_n$ that get mapped to zero. We therefore define $f_u$ to be the monotonic Boolean function that maps the set $U$ to zero such that that $f_U \leq f_V \iff V \subseteq U$. The topmost function on this lattice is $f_\emptyset$, and below it are functions $f_A$ where $A$ corresponds to some subset of $\mathcal{B}_n$. When $A$ contains only singletons, then $f_A$ is part of the Boolean algebra between $f_\emptyset$ and $f_{\{\{x\}\mid x \in S\}}$. However, $A$ might also contain non-singleton elements from $\mathcal{B}_n$. Assume that the interval $[f_A, f_\emptyset]$ is a Boolean algebra, and that $A$ contains non-singleton elements. Let $b\in A$ denote the non-singleton element. Then construct $f_{A.b} \coloneq f_{A \setminus \{\{x\} \mid x \in b\}}$. Since $[f_A, f_\emptyset]$ is a Boolean algebra, every $f_{A.b}$ is guaranteed to exist. However, $f_{A.b}$ maps all elements of $b$ to 1, but $b$ itself to 0. This violates monotonicity, so $f_{A.b}$ is not part of the lattice. This is a contradiction and shows that the only elements that are part of the upper Boolean algebra are the $f_A$ where $A$ contains only singletons. Since there is thus only one boolean algebra on $n$ variables at the top, this must be the same as the one from i). This proves ii). 

    iii) Finally, we show that $\mu_{\mathcal{R}_n}(\bm\alpha, \hat{1})\neq 0 \iff \bm\alpha$ is part of the upper Boolean algebra. It is known that Equation \eqref{eq:MF_distributive} can be written as \cite{stanley2011enumerative}:
    \begin{align}
    \mu_{\mathcal{R}_n}(A, B) = \begin{cases}
            (-1)^{|I_B \setminus I_A|}& \text{if $[I_A, I_B]$ is a Boolean algebra,}\\
            0 & \text{otherwise.}
        \end{cases}
    \end{align}
    The complete synergy atom can be written as
    \begin{align}
        I_\partial^{(1\ldots n)}(S, T) = \sum_{\bm\alpha \leq (1\ldots n)} \mu_{\mathcal{R}_n}(\bm\alpha, (1\ldots n)) I_\cap^{\bm\alpha}(\bm X; Y),
        \intertext{which reduces to}
        = \sum_{\bm\alpha: [\bm\alpha, (1\ldots n)] \text{ is a Boolean algebra}} (-1)^{|I_{\bm\alpha^*} \setminus I_\emptyset|}I_\cap^{\bm\alpha}(\bm X; Y).
    \end{align}
    By i) and ii), the only antichains $a$ for which $[a, (1\dots n)]$ is a Boolean algebra are the redundancies among $(n-1)$-variable synergies. This means that
      \begin{align}
            I_\partial^{(1\dots n)}(\bm X; Y) = \sum_{U \subseteq \{1,\dots,n\}} (-1)^{n-|U|} I^{S_n^U}_\cap(\bm X; Y),
        \end{align}
        where $S_n^U = \{\{1, \ldots,  n\}\setminus x \mid x \in U\}$ and completes the proof. 
\end{proof}

\section{Case studies}

\subsection{Pre-processing of music scores}

Our analyses used electronic scores publicly available at 
\href{http://kern.ccarh.org}{\texttt{http://kern.ccarh.org}}, focusing
on scores with four melodic lines: four voices (soprano,
alto, tenor and bass) in the case of Bach’s chorales,
and four string instruments (1st violin, 2nd violin, viola and cello) in the case of Corelli’s pieces. The scores
were pre-processed in Python using the \texttt{Music21} package\footnote{\href{http://web.mit.edu/music21}{\texttt{http://web.mit.edu/music21}}}, which allowed us to select only the pieces written in major mode, and to transpose them to C major. The melodic lines were transformed into time series of 13 possible values (one for each
note plus one for the silence), using the smallest rhythmic
duration as time unit. This generated $\approx 4\times10^4$ four-note chords for the chorales, and $\approx 8\times10^4$ for Corelli’s pieces. 
Chords that persist for multiple measures show up as repeated entries in the time series; for our calculations, we removed any consecutive copies of a chord, keeping only the chord \textit{changes}. We estimated the probability of observing one chord transitioning to another simply by counting the number of such transitions in the repertoire and normalising this by the total number of chord changes.

\subsection{Additional analyses on Bach's Chorales \label{app:bach}}

While both composers show more downward than upward causation, indicating strong harmonic constraints, Bach scores significantly higher in both categories. To confirm that the atoms indeed capture the musical dependencies, we modified the Bach chorales and verified that the atoms change as expected. The first perturbed version, referred to as 2-fold redundant Bach, copied the Soprano and Alto of the original, but replaced the Tenor and Bass. Let $X_t \in \Sigma$ denote the value of voice $X$ at timestep $t$. We then set $T_{t+1} = (S_t - 4)\mod 13$ and $B_{t+1} = (S_t - 6) \mod 13$. The 3-fold redundant Bach dataset is then formed by in addition setting $A_{t+1} = (S_t - 2) \mod 13$. We then indeed find that the $\Phi$ID atoms change as expected, shown in Table \ref{tab:perturbed_bach}.
\begin{table}[h]
    \centering
    \caption{}
    \begin{tabular}{r||l|l|l}
        & Original & 2-fold Red. & 3-fold Red.\\
        \hline
        $S \to T|B$     & 0.00 & 1.45 & 0.00 \\
         $S \to A|T|B$  & 0.00 & 0.00 & 1.94
    \end{tabular}
    \label{tab:perturbed_bach}
\end{table}
\newpage

As an additional control, we randomly shuffled the chords from Bach's chorales and compared the mean value of atoms in each of the information dynamics categories with those of the unshuffled Bach. Figure \ref{fig:infoDyn_bach_v_shuf} shows that the Bach chorales showed significant information dynamics in all categories, except for Copy and Erasure. Results for Corelli were identical. This null distribution additionally shows that the effect of bias-correction is negligible, as the shuffled atoms are two to three orders of magnitude weaker than the unshuffled atoms.

\begin{figure}[t!]
    \centering
    \includegraphics[width=0.46\textwidth]{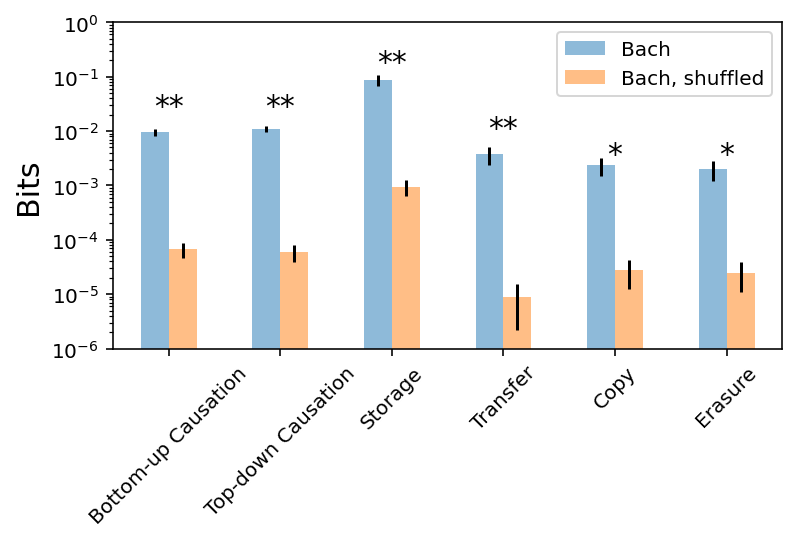}
    \caption{The mean absolute value of atoms in Bach chorales and a shuffled version in each of the six categories of information dynamics. Only in the first four categories does the original Bach music score significantly different from a randomly shuffled version after Bonferroni correction. 
    Significance was calculated via a t-test, 
    where * means $p\leq 0.05$ and ** $p\leq \frac{0.05}{6}$ (that is, ** indicates significance after Bonferroni correction). Error bars indicate standard error of the mean.}
    \label{fig:infoDyn_bach_v_shuf}
\end{figure}

\newpage

\bibliography{references}
\bibliographystyle{apsrev4-1}

\end{document}